\documentclass[11pt]{article}

\usepackage[utf8]{inputenc}
\usepackage[T1]{fontenc}
\usepackage{lmodern}
\usepackage[a4paper,margin=1in]{geometry}
\usepackage{amsmath,amssymb,amsthm}
\usepackage{bm}
\usepackage{graphicx}
\usepackage{xcolor}
\usepackage{booktabs}
\usepackage{array}
\usepackage{multirow}
\usepackage{tabularx}
\usepackage{subcaption}
\usepackage{algorithm}
\usepackage{algpseudocode}
\usepackage{enumitem}
\usepackage[title]{appendix}
\usepackage[authoryear,round]{natbib}
\usepackage{microtype}
\usepackage{hyperref}
\hypersetup{hidelinks}

\theoremstyle{plain}
\newtheorem{prediction}{Prediction}
\newtheorem{lemma}{Lemma}

\newcommand{\Xup}{\ensuremath{\overline{X}}}
\newcommand{\Xlow}{\ensuremath{\underline{X}}}
\newcommand{\Yup}{\ensuremath{\overline{Y}}}
\newcommand{\Ylow}{\ensuremath{\underline{Y}}}
\newcommand{\Ssim}{\ensuremath{S}}
\newcommand{\Sdis}{\ensuremath{\overline{S}}}
\newcommand{\Lset}{\ensuremath{\mathcal{L}}}
\newcommand{\Cset}{\ensuremath{\mathcal{C}}}
\newcommand{\Rpref}{\ensuremath{R_{\mathrm{pref}}}}
\newcommand{\Rorth}{\ensuremath{R_{\mathrm{orth}}}}
\newcommand{\xtilde}{\tilde{x}}
\DeclareMathOperator*{\argmax}{arg\,max}

\newcommand{\figplaceholder}[2]{%
  \fbox{\begin{minipage}[c][#1][c]{0.92\linewidth}
    \centering\itshape\small\color{gray}
    [\,FIGURE PLACEHOLDER --- replace with image\,]\\[2pt]
    #2
  \end{minipage}}%
}

\newcommand{\safeincludegraphics}[3]{%
  \IfFileExists{#2}{%
    \includegraphics[#1]{#2}%
  }{%
    \figplaceholder{#3}{Missing file: \texttt{\detokenize{#2}}}%
  }%
}

\begin{document}

%------------------------------------------------------------------------------
% Title & Authors Block (Standalone Article Format)
%------------------------------------------------------------------------------
\title{A Bayes--Markov Neuromorphic Model of Cortical Orientation Selectivity: A Computational Re-implementation and Quantitative Simulation Study}

\author{
Abolfazl Moslemi\thanks{Corresponding author: \texttt{abolfazl.moslemi14@sharif.edu}}\\
Department of Mathematical Sciences, Sharif University of Technology, Tehran, Iran
\and
Milad Sarabadani\\
School of Mathematics, Statistics, and Computer Science, Faculty of Science,\\
University of Tehran, Tehran, Iran
\and
Fatemeh Sefidian\\
Department of Psychology, Islamic Azad University, Tehran, Iran
\and
Hossein Peyvandi\\
Pasargad Institute for Advanced Innovative Solutions (PIAIS), Tehran, Iran
}

\date{}

\maketitle

%------------------------------------------------------------------------------
% Abstract
%------------------------------------------------------------------------------
\begin{abstract}
The emergence of orientation selectivity in the primary visual cortex (V1) remains a central question in computational neuroscience. Shirazi's Bayes--Markov model proposed a probabilistic explanation for how orientation-selective inhibition can arise from non-oriented lateral geniculate nucleus (LGN) inputs through local inference. In that formulation, the activity pattern of striate cortical inhibitory (SCI) cells is estimated from the LGN activity pattern by a maximum a posteriori (MAP) criterion over a two-layer hierarchical Markov random field, and the resulting inference is implemented through a local parallel relaxation algorithm. We provide a computationally explicit re-implementation and quantitative simulation study of this framework. We reconstruct the mathematical model, describe its fully LGN-driven update rule, and implement a vectorized simulation framework that preserves the original local clique operations while making systematic parameter sweeps feasible. We evaluate the model using orientation tuning curves, an orientation selectivity index (OSI), controlled LGN noise perturbations, contrast tests, and model-variant comparisons. We further add a spiking SCI-layer realization using leaky integrate-and-fire and Hodgkin--Huxley neurons to examine whether the rate-coded SCI field can be expressed through temporally explicit neural activity. The simulations support the central qualitative behavior of the Bayes--Markov framework: sharp orientation selectivity, robustness to moderate LGN noise, and a biologically interpretable proof-of-concept spiking realization of the inferred inhibitory field.
\end{abstract}

\noindent\textbf{Keywords:} Primary visual cortex; Orientation selectivity; Markov random fields; Bayesian inference; Neuromorphic modeling; Spiking neurons

%==============================================================================
\section{Introduction}
\label{sec:intro}
%==============================================================================

Ever since the seminal observations of Hubel and Wiesel \citep{hubel1959,hubel1962,hubel1977}, orientation selectivity in the primary visual cortex has been a central problem in visual neuroscience and cortical computation \citep{marr1982,ferster2000,pugh2000}. Their classical feedforward account proposed that cortical orientation tuning arises from the spatial alignment of thalamocortical afferents. In this geometric view, a cortical neuron becomes selective for an oriented stimulus because it receives excitatory input from LGN cells whose receptive fields are arranged along the preferred orientation.

Although this explanation is elegant, physiological and computational studies have shown that feedforward geometry alone does not account for every observed property of V1 tuning. Contrast-invariant tuning, recurrent modulation, normalization, and inhibitory interactions have therefore motivated models in which local excitation and inhibition participate in the formation, amplification, stabilization, or sharpening of orientation-selective responses \citep{sclar1982,hata1988,benyishai1995,somers1995,troyer1998,marino2005,palmer2007,carandini2012}.

This leads to a central bootstrapping problem. If orientation-specific inhibition contributes to cortical orientation selectivity, how can the inhibitory network itself become orientation selective when its thalamic inputs are assumed to be non-oriented? \citet{shirazi2004} addressed this problem by formulating a Bayes--Markov computational model in which orientation-selective inhibition emerges from local cooperative inference over non-oriented LGN activity. The model represents the LGN and SCI layers as two coupled random fields and estimates the SCI activity pattern from the observed LGN pattern using a MAP criterion. The resulting computation is performed by a deterministic, local, parallel relaxation algorithm grounded in established Markov-random-field methods \citep{besag1974,geman1984,besag1986}.

The present study revisits this Bayes--Markov model in a computationally explicit form. The original probabilistic formulation and biological interpretation are preserved, while the implementation, numerical evaluation, and spiking readout are made systematic and reproducible. The contribution is therefore a quantitative computational reassessment of the original framework rather than a replacement biological theory.

\paragraph{Contributions of the present study.}
The main contributions of this work are as follows:
\begin{enumerate}[leftmargin=1.8em,itemsep=3pt]
    \item We reconstruct the Bayes--Markov formulation of Shirazi's model, including the two-layer HMRF structure, the MAP objective, and the local GICM-style relaxation rule.
    \item We describe the fully LGN-driven version of the model, emphasizing how the original formulation avoids biologically implausible global normalization and fixed reference signals.
    \item We provide a vectorized computational implementation that preserves the original local clique-energy operations while making systematic simulation studies feasible.
    \item We add a quantitative layer to the model's previously qualitative evaluation: orientation tuning curves with an orientation selectivity index (OSI), contrast invariance, and robustness to LGN noise across all model variants.
    \item We extend the simulation study with a spiking SCI-layer implementation using LIF and Hodgkin--Huxley neurons, showing that the inferred orientation-selective field can be expressed through temporally explicit spike trains.
\end{enumerate}

%==============================================================================
\section{Related Models of Cortical Orientation Selectivity}
\label{sec:related}
%==============================================================================

A wide range of models has been proposed to explain orientation selectivity in V1. Classical feedforward models attribute tuning to the spatial arrangement of LGN receptive fields, whereas recurrent models emphasize local excitation, inhibition, balance, and gain control \citep{ferster2000,ursino2004,marino2005}. More recent experimental and computational work has continued to examine how synaptic organization, random or structured recurrent connectivity, and layered spiking dynamics shape selectivity \citep{sadeh2014,sadeh2015,wilson2016,merkt2019,rossi2020,chariker2020}.

\paragraph{Feedforward and geometric models.}
The classical Hubel--Wiesel model explains orientation selectivity through the alignment of thalamic afferents \citep{hubel1962}. Related geometric models provide simple and intuitive mechanisms for generating orientation-biased responses. However, threshold nonlinearities can produce an iceberg effect in which apparent tuning width changes with stimulus contrast. This is important because V1 orientation tuning is often approximately contrast invariant \citep{sclar1982,troyer1998,ferster2000,palmer2007}.

\paragraph{Recurrent cortical models.}
Recurrent models emphasize the role of intracortical circuitry in shaping and stabilizing orientation tuning. In such models, orientation selectivity is amplified, sharpened, balanced, or stabilized through interactions among cortical neurons \citep{benyishai1995,somers1995,marino2005,sadeh2015}. These approaches account for several response properties but differ in the assumed structure and functional role of excitation and inhibition \citep{ursino2004,carandini2012}. Experimental studies likewise indicate that orientation selectivity depends on the organization of synaptic and recurrent inputs rather than on a single universal circuit motif \citep{wilson2016,rossi2020}.

\paragraph{Bayesian and probabilistic accounts.}
Probabilistic inference has also been used as a framework for interpreting cortical computation, including predictive and hierarchical processing in the visual system \citep{rao1999,lee2003}. These accounts provide a broader computational context for local MAP-based models in which hidden neural states are inferred from observed sensory activity.

\paragraph{Bayes--Markov model.}
The Bayes--Markov model of \citet{shirazi2004} is distinctive because it addresses the bootstrapping problem directly. It asks how orientation-selective inhibition itself can emerge from non-oriented LGN inputs. The model treats the SCI layer as a hidden random field and the LGN layer as an observed random field. The cortical circuit is interpreted as performing local Bayesian inference to estimate the hidden SCI configuration from the observed LGN pattern.

\begin{table}[htpb]
\centering
\caption{Conceptual comparison of selected model classes for cortical orientation selectivity.}
\label{tab:model_comparison_conceptual}
\renewcommand{\arraystretch}{1.2}
\begin{tabularx}{\textwidth}{@{}p{0.23\textwidth}X p{0.18\textwidth}p{0.18\textwidth}@{}}
\toprule
\textbf{Model class} & \textbf{Main mechanism} & \textbf{Intracortical role} & \textbf{Bootstrapping addressed?} \\
\midrule
Feedforward geometric models & Aligned LGN inputs & Limited & No \\
Recurrent cortical models & Local excitation/inhibition & Strong & Partly \\
Balanced random networks & Dynamic balance & Strong & Partly \\
Bayes--Markov model & Local MAP inference & Essential & Yes \\
\bottomrule
\end{tabularx}
\end{table}

%==============================================================================
\section{Probabilistic Framework}
\label{sec:framework}
%==============================================================================

We now summarize the probabilistic formulation of the Bayes--Markov model. The purpose of this section is not to introduce a new mathematical model, but to restate the original formulation in a notation suitable for computational implementation.

\subsection{Problem Formulation}

Let
\[
\Lset=\{(i,j): i,j=1,2,\dots,N\}
\]
denote the set of sites on an $N\times N$ lattice. The LGN layer and SCI layer are assumed to be topographically aligned on this lattice. Let
\[
Y_{\Lset}=\{Y_{ij}:(i,j)\in\Lset\}
\]
denote the LGN random field and
\[
X_{\Lset}=\{X_{ij}:(i,j)\in\Lset\}
\]
denote the SCI random field. The LGN cells are assumed to be non-oriented, whereas the SCI cells are assumed to be tuned to a particular orientation, for example horizontal orientation.

Each SCI variable takes one of two states,
\[
X_{ij}\in Q^X_{ij}=\{\Xlow,\Xup\},
\]
where $\Xlow$ denotes the less active or silent state and $\Xup$ denotes the active state. Similarly, the LGN activity is represented by a discrete local state space,
\[
Y_{ij}\in Q^Y_{ij}=\{\Ylow_{ij},\Yup_{ij}\}.
\]
The values $\Ylow_{ij}$ and $\Yup_{ij}$ may vary from site to site, reflecting the fact that LGN firing-rate ranges depend on the local visual input.

The computational goal is to estimate the SCI activity pattern from the observed LGN activity pattern. For a horizontally tuned SCI layer, the desired computation is a spatial filtering operation: horizontal line-like LGN patterns should activate corresponding SCI cells, while dots, blobs, and non-horizontal structures should be suppressed.

%------------------------------------------------------------------------------
\subsection{Hierarchical Markov Random Field Model}
%------------------------------------------------------------------------------

The joint random field $(X_{\Lset},Y_{\Lset})$ is modeled through a two-layer hierarchical Markov random field. The SCI layer is treated as a hidden MRF, and the LGN layer is treated as an observed field whose conditional distribution depends on the SCI state, following the Gibbs--MRF formulation used in spatial statistical modeling and image inference \citep{besag1974,geman1984,noda1995}.

The prior distribution of the SCI field is written as a Gibbs distribution,
\begin{equation}
P(x_{\Lset})
=
Z_X^{-1}
\exp\left[
-
\sum_{(i,j)\in\Lset}
\sum_{C\in\Cset^{X}_{ij}}
E(x_C)
\right],
\label{eq:gibbs_prior}
\end{equation}
where $\Cset^X_{ij}$ is the set of local SCI cliques associated with site $(i,j)$, $E(x_C)$ is the clique energy, and
\begin{equation}
Z_X
=
\sum_{x_{\Lset}\in\Omega_X}
\exp\left[
-
\sum_{(i,j)\in\Lset}
\sum_{C\in\Cset^{X}_{ij}}
E(x_C)
\right]
\label{eq:partition_x}
\end{equation}
is the partition function.

The conditional distribution of the LGN field given the SCI field is similarly written as
\begin{equation}
P(y_{\Lset}\mid x_{\Lset})
=
Z_{Y\mid X}^{-1}
\exp\left[
-
\sum_{(i,j)\in\Lset}
\sum_{C\in\Cset^{Y}_{ij}}
E(y_C\mid x_{ij})
\right],
\label{eq:conditional_lgn}
\end{equation}
where $\Cset^Y_{ij}$ is the set of LGN cliques associated with site $(i,j)$.

The hidden SCI field encourages spatially coherent active regions, while the conditional LGN field encodes orientation-specific local structure. Thus, orientation selectivity is not imposed by a rigid feedforward alignment of LGN inputs; it emerges from the local energy structure used to infer the hidden SCI configuration.

%------------------------------------------------------------------------------
\subsection{MAP Objective}
%------------------------------------------------------------------------------

Given an observed LGN pattern $y_{\Lset}$, the MAP estimate of the SCI field is
\begin{equation}
x_{\Lset}^{\ast}
=
\argmax_{x_{\Lset}\in\Omega_X}
P(x_{\Lset}\mid y_{\Lset}).
\label{eq:map1}
\end{equation}
Using Bayes' rule and ignoring the evidence term $P(y_{\Lset})$, this is equivalent to
\begin{equation}
x_{\Lset}^{\ast}
=
\argmax_{x_{\Lset}\in\Omega_X}
P(y_{\Lset}\mid x_{\Lset})P(x_{\Lset}).
\label{eq:map2}
\end{equation}
Because the state space has size $2^{N^2}$ for a binary $N\times N$ SCI lattice, exhaustive global optimization is not biologically or computationally feasible. The original model therefore uses a local parallel relaxation algorithm to obtain a suboptimal but practical MAP estimate, closely related to iterated conditional modes and deterministic HMRF relaxation methods \citep{besag1986,shirazi1993}.

%------------------------------------------------------------------------------
\subsection{Local Parallel Relaxation}
%------------------------------------------------------------------------------

The local update rule is a generalized iterative conditional modes procedure \citep{besag1986}. At iteration $n+1$, each SCI site is updated by maximizing a local objective:
\begin{align}
\hat{x}_{ij}(n+1)
=
\argmax_{x_{ij}\in Q^X_{ij}}
\Bigg\{
&-
\sum_{C\in\Cset^Y_{ij}}
E(y_{ij},y_{C'}\mid x_{ij})
\nonumber\\
&-
\ln
\sum_{k\in Q^Y_{ij}}
\exp\left[
-
\sum_{C\in\Cset^Y_{ij}}
E(k,y_{C'}\mid x_{ij})
\right]
\nonumber\\
&-
\sum_{C\in\Cset^X_{ij}}
E(x_{ij},\hat{x}_{C'}(n))
\Bigg\},
\label{eq:gicm}
\end{align}
where $C'=C-\{(i,j)\}$ denotes the neighboring site or sites in clique $C$ excluding the current site. This rule is local because each update depends only on the LGN and SCI states in a small neighborhood around $(i,j)$.

%==============================================================================
\section{Fully LGN-Driven Bayes--Markov Model}
\label{sec:model}
%==============================================================================

Although Eq.~\eqref{eq:gicm} provides a computational solution to the MAP problem, it contains operations that are difficult to interpret as biologically plausible. In particular, the logarithmic normalizer requires knowledge of possible LGN activity levels and involves a nonlinear log-sum-exp operation. \citet{shirazi2004} therefore transformed the update into a fully LGN-driven and more physiologically plausible form.

In the fully LGN-driven version, the update rule can be written schematically as
\begin{align}
\xtilde_{ij}(n+1)
=
\argmax_{x_{ij}\in\{\Xlow,\Xup\}}
\Bigg\{
&-
\sum_{C\in\Cset^Y_{ij}}
E(y_{ij},y_{C'}\mid x_{ij};\beta,T)
\nonumber\\
&-
\max
\left\{
-
\sum_{C\in\Cset^Y_{ij}}
E(y_{ij},y_{C'}\mid x_{ij};\beta,T),
\;
-
\sum_{C\in\Cset^Y_{ij}}
\overline{E}(y_{ij},y_{C'}\mid x_{ij};\beta,T)
\right\}
\nonumber\\
&-
\sum_{C\in\Cset^X_{ij}}
E(x_{ij},\xtilde_{C'}(n);\gamma(n))
\Bigg\},
\label{eq:update}
\end{align}
where the complementary energy is
\[
\overline{E}(y_{ij},y_{C'}\mid x_{ij};\beta,T)
=
-
E(y_{ij},y_{C'}\mid x_{ij};\beta,T).
\]
The parameter $\beta>0$ controls the strength of the LGN clique energies, $T$ is a local similarity threshold, and $\gamma(n)$ controls the recurrent SCI coupling strength at iteration $n$.

A common choice for the time-varying SCI coupling is
\begin{equation}
\gamma(n)=\bar{\gamma}\tanh(n),
\label{eq:gamma_schedule}
\end{equation}
which allows the model to begin with weak recurrent coupling and then gradually enforce local SCI consistency.

%------------------------------------------------------------------------------
\subsection{Orientation-Specific LGN Clique Energies}
%------------------------------------------------------------------------------

To define a horizontally tuned model, the local clique set is divided into horizontal and non-horizontal doubleton cliques. Let $\Cset^{H,Y}_{ij}$ denote the horizontal LGN cliques and let $\Cset^{\overline{H},Y}_{ij}$ denote the non-horizontal LGN cliques.

For a horizontal doubleton clique $C\in\Cset^{H,Y}_{ij}$, the conditional energy is defined as
\begin{align}
E(y_{ij},y_{C'}\mid \Xup;\beta,T)
&=
-\beta\,\Ssim(y_{ij},y_{C'};T)
+
\beta\,\Sdis(y_{ij},y_{C'};T),
\label{eq:horizontal_active}
\\
E(y_{ij},y_{C'}\mid \Xlow;\beta,T)
&=
\beta\,\Ssim(y_{ij},y_{C'};T)
-
\beta\,\Sdis(y_{ij},y_{C'};T).
\label{eq:horizontal_inactive}
\end{align}
For a non-horizontal clique $C\in\Cset^{\overline{H},Y}_{ij}$, the signs are reversed:
\begin{align}
E(y_{ij},y_{C'}\mid \Xup;\beta,T)
&=
\beta\,\Ssim(y_{ij},y_{C'};T)
-
\beta\,\Sdis(y_{ij},y_{C'};T),
\label{eq:nonhorizontal_active}
\\
E(y_{ij},y_{C'}\mid \Xlow;\beta,T)
&=
-\beta\,\Ssim(y_{ij},y_{C'};T)
+
\beta\,\Sdis(y_{ij},y_{C'};T).
\label{eq:nonhorizontal_inactive}
\end{align}

The functions $\Ssim$ and $\Sdis$ are local similarity and dissimilarity operators:
\begin{equation}
\Ssim(y_1,y_2;T)
=
\begin{cases}
1, & |y_1-y_2|\le T,\\
0, & |y_1-y_2|>T,
\end{cases}
\label{eq:similarity}
\end{equation}
and
\begin{equation}
\Sdis(y_1,y_2;T)
=
1-\Ssim(y_1,y_2;T).
\label{eq:dissimilarity}
\end{equation}
Thus, horizontal neighboring LGN cells with similar activity are encouraged to support an active SCI state, while non-horizontal similarity supports the inactive SCI state.

%------------------------------------------------------------------------------
\subsection{SCI Clique Energies}
%------------------------------------------------------------------------------

The SCI prior is isotropic and encourages local spatial continuity. Its doubleton-clique energy is defined as
\begin{equation}
E(x_{ij},\xtilde_{C'}(n);\gamma(n))
=
-\gamma(n) I(x_{ij},\xtilde_{C'}(n)),
\label{eq:esci}
\end{equation}
where
\begin{equation}
I(a,b)
=
\begin{cases}
1, & a=b,\\
0, & a\ne b.
\end{cases}
\label{eq:identity}
\end{equation}
Unlike the LGN clique energies, this SCI term does not encode a preferred orientation. Instead, it promotes spatially coherent cortical activity patterns after the orientation-specific evidence has been extracted from the LGN layer.

%------------------------------------------------------------------------------
\subsection{Inference Algorithm}
\label{sec:inference_algorithm}
%------------------------------------------------------------------------------

The MAP estimate of the SCI field is computed by an iterated conditional modes (ICM)
procedure with two phases. In the first phase, each site is initialized from the LGN
evidence alone using a maximum-likelihood (ML) estimate, in which the lateral coupling is
inactive ($\gamma(0)=0$). In the second phase, the estimate is refined over successive
iterations by the full MAP update of Eq.~\eqref{eq:update}, which adds the recurrent SCI
coupling whose strength $\gamma(n)$ increases with the iteration index according to
Eq.~\eqref{eq:gamma_schedule}. The two phases are summarized in
Algorithm~\ref{alg:inference}.

\begin{algorithm}[htpb]
\caption{Fully LGN-Driven Bayes--Markov SCI Inference}
\label{alg:inference}
\begin{algorithmic}[1]
\Procedure{InferSCI}{$Y$, $\beta$, $\gamma_0$, $T$, $N_{\mathrm{iter}}$}
  \State \textbf{Input:} LGN activity field $Y=\{y_{ij}\}$; parameters $\beta,\gamma_0,T$; iterations $N_{\mathrm{iter}}$
  \State \textbf{Output:} SCI field $X=\{x_{ij}\}$, $x_{ij}\in\{\Xlow,\Xup\}$
  \Statex
  \Statex \textit{Phase 1 --- Maximum-likelihood initialization ($\gamma=0$)}
  \For{each site $(i,j)$ in parallel}
    \State $E_{\mathrm{a}} \gets \sum_{C\in\Cset^{Y}_{ij}} E(y_{ij},y_{C'}\mid \Xup;\beta,T)$
           \Comment{evidence for active state}
    \State $E_{\mathrm{s}} \gets \sum_{C\in\Cset^{Y}_{ij}} E(y_{ij},y_{C'}\mid \Xlow;\beta,T)$
           \Comment{evidence for inactive state}
    \State $x_{ij} \gets \Xup$ \textbf{if} $-E_{\mathrm{a}} > -E_{\mathrm{s}}$ \textbf{else} $\Xlow$
  \EndFor
  \Statex
  \Statex \textit{Phase 2 --- MAP refinement with recurrent coupling}
  \For{$n \gets 1$ \textbf{to} $N_{\mathrm{iter}}$}
    \State $\gamma(n) \gets \gamma_0\,\Gamma(n)$ \Comment{coupling schedule, $\Gamma(n)\!\to\!1$}
    \State $\tilde{X} \gets X$ \Comment{freeze neighbors for this sweep}
    \For{each site $(i,j)$ in parallel}
      \State $U_{\mathrm{a}} \gets -E_{\mathrm{a}}(i,j) - \!\!\sum_{C\in\Cset^{X}_{ij}}\!\! E(\Xup,\tilde{x}_{C'};\gamma(n))$
      \State $U_{\mathrm{s}} \gets -E_{\mathrm{s}}(i,j) - \!\!\sum_{C\in\Cset^{X}_{ij}}\!\! E(\Xlow,\tilde{x}_{C'};\gamma(n))$
      \State $x_{ij} \gets \Xup$ \textbf{if} $U_{\mathrm{a}} > U_{\mathrm{s}}$ \textbf{else} $\Xlow$
    \EndFor
  \EndFor
  \State \Return $X$
\EndProcedure
\end{algorithmic}
\end{algorithm}

The procedure is fully parallel: every site is updated from the frozen neighbor field of
the previous sweep, so the update order does not affect the result within an iteration.
The LGN evidence terms $E_{\mathrm{a}}$ and $E_{\mathrm{s}}$ depend only on the fixed input
field $Y$ and are therefore computed once and reused across all MAP iterations; only the
recurrent SCI term is recomputed per iteration. Empirically, the procedure converges within
a small number of iterations, consistent with the fast convergence reported for the
original model.

The model variants differ only in the normalization used for the evidence terms and in the
coupling schedule. The biologically motivated variant (V1ModelBio) uses a local maximum
operation in place of a log-sum-exp normalization and a smoothly increasing coupling
$\Gamma(n)=\tanh(n)$, whereas the baseline variant (V1Model) uses the exact log-sum-exp
normalization with a fixed coupling.

%==============================================================================
\section{Neuromorphic Interpretation}
\label{sec:neuromorphic}
%==============================================================================

The Bayes--Markov model is not only a statistical model; it was also proposed as a computationally plausible account of how local cortical circuitry could implement orientation-selective inhibition. This interpretation is consistent with broader probabilistic and predictive accounts in which recurrent cortical interactions implement local inference over latent causes \citep{rao1999,lee2003}. The operations in Eq.~\eqref{eq:update} can be interpreted through interacting functional layers (Fig.~\ref{fig:functional_layers}).

First, a contrast-detection stage computes local similarity and dissimilarity between neighboring LGN firing rates. Second, a normalization stage implements the local maximum operation used to avoid biologically implausible log-sum-exp normalization. Third, a hypothesis-testing stage compares the evidence for the active and inactive SCI states. Finally, a regularization stage implements the isotropic SCI coupling through local recurrent interactions.

\begin{figure}[htpb]
\centering
\safeincludegraphics{width=0.82\linewidth}{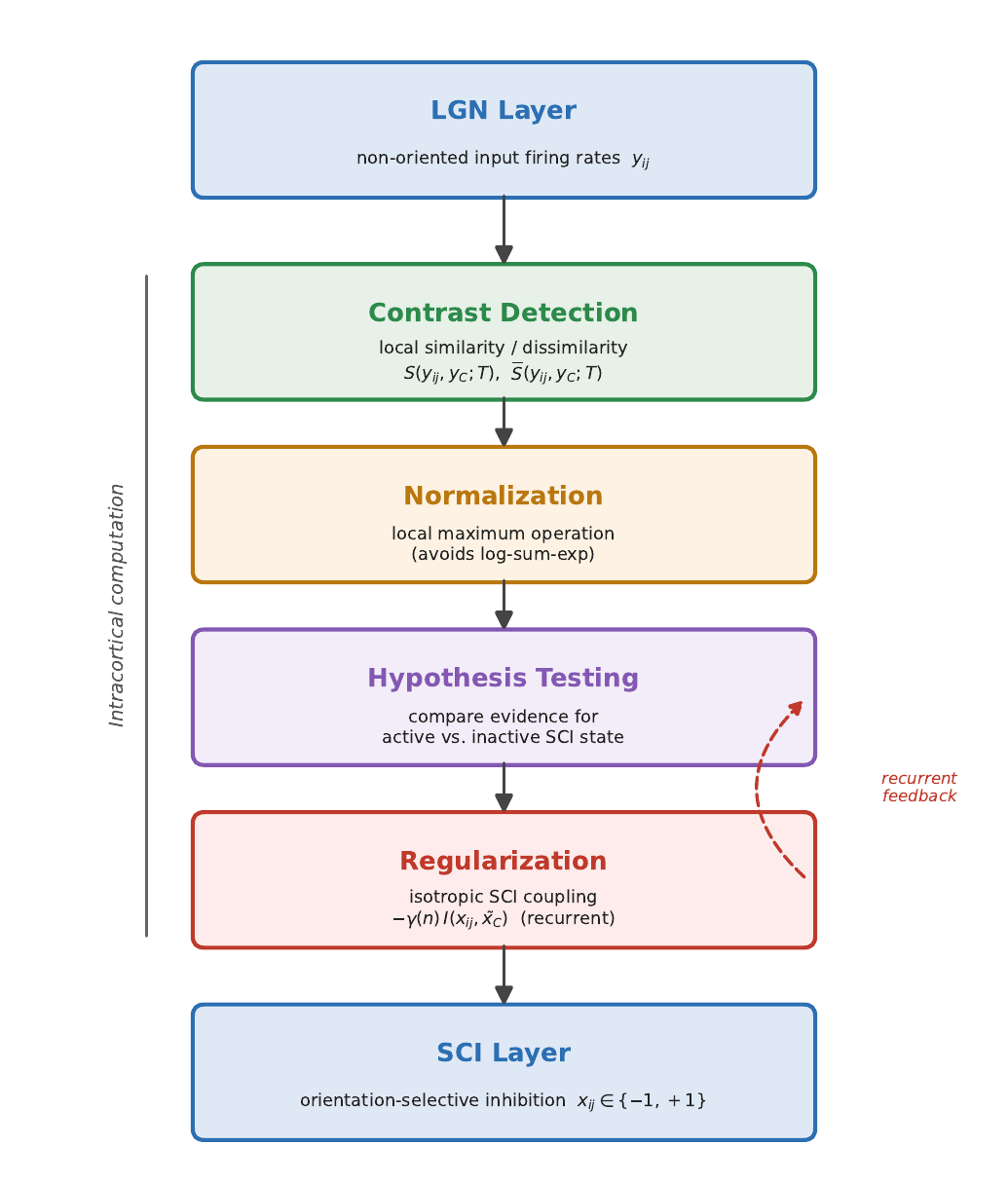}{6.0cm}
\caption{
Conceptual neuromorphic interpretation of the Bayes--Markov computation. The model can be viewed as a sequence of local operations: contrast detection, normalization, hypothesis testing, and recurrent SCI regularization.
}
\label{fig:functional_layers}
\end{figure}

This interpretation is important because it shows that the model does not rely on a global controller or a rigid feedforward geometric template. Orientation-selective inhibition is produced through local cooperative interactions that estimate the hidden SCI field from the observed LGN field.

%==============================================================================
\section{Computational Simulation Framework}
\label{sec:experiments}
%==============================================================================

The goal of the simulation study is to provide a quantitative computational evaluation of the original Shirazi model. The mathematical formulation, local clique-energy definitions, MAP objective, and biological interpretation are kept fixed. The contribution is a systematic assessment of the model's behavior under controlled stimuli, noise perturbations, model comparisons, and spiking SCI-layer realizations.

All simulations were performed on $128\times128$ lattices. Unless otherwise specified, the parameters were fixed at
\[
\beta=1.0,\qquad \gamma=2.0,\qquad T=0.15.
\]
The main rate-coded implementation is referred to as V1ModelBio; we also include a baseline rate-coded variant, V1Model. These two rate-coded variants are compared against the spiking realizations (LIF and HH) introduced in Section~\ref{sec:spiking}.

%------------------------------------------------------------------------------
\subsection{Vectorized Implementation and Correctness Verification}
%------------------------------------------------------------------------------

A direct implementation of the model updates each pixel by explicitly looping over the $128\times128$ lattice and evaluating local clique contributions. This is faithful to the mathematical model but slow for repeated simulations. To make orientation sweeps, noise sweeps, and model comparisons feasible, the local operations were reimplemented in vectorized form using shifted arrays and padded-neighborhood extraction.

This vectorization does not change the model. It computes the same local clique energies and the same update rule, but it evaluates them over the entire lattice using array operations. Correctness was verified by comparing the vectorized implementation with the original pixel-wise implementation on controlled test inputs. The outputs matched exactly, with zero mismatched pixels both after the maximum-likelihood initialization step and after the full forward pass. Therefore, the quantitative experiments below use a faster implementation that is equivalent to the original local update rules.

%------------------------------------------------------------------------------
\subsection{LGN Input Construction}
%------------------------------------------------------------------------------

For controlled synthetic experiments, LGN activity patterns were generated from bar stimuli, line-pair stimuli, and noisy versions of these inputs. When an image-based input is used, the raw grayscale image is first converted into a contrast-sensitive LGN activation field. The preprocessing is intended as a computational approximation to center-surround contrast processing in the retinal ganglion--LGN pathway.

\begin{algorithm}[htpb]
\caption{Synthesizing LGN Activation Fields}
\label{alg:lgn}
\begin{algorithmic}[1]
\Require Grayscale input image $I(x,y)$, smoothing variance $\sigma$, small constant $\epsilon$
\Ensure LGN activity matrix $Y$
\State $I_{\mathrm{blur}} \gets I * G_{\sigma}$ \Comment{Gaussian spatial integration}
\State $S_x \gets \nabla_x I_{\mathrm{blur}}$ \Comment{Horizontal Sobel derivative}
\State $S_y \gets \nabla_y I_{\mathrm{blur}}$ \Comment{Vertical Sobel derivative}
\State $M \gets \sqrt{S_x^2+S_y^2}$ \Comment{Local contrast magnitude}
\State $Y \gets \dfrac{M-\min(M)}{\max(M)-\min(M)+\epsilon}$ \Comment{Normalized LGN activity}
\State \Return $Y$
\end{algorithmic}
\end{algorithm}

%------------------------------------------------------------------------------
\subsection{Orientation Selectivity Index}
%------------------------------------------------------------------------------

To quantify tuning sharpness, we use the preferred-versus-orthogonal orientation selectivity index employed in the original simulation analysis. Let $\Rpref$ denote the mean response at the preferred orientation and $\Rorth$ denote the mean response at the orthogonal orientation. This two-point index complements full-curve measures such as circular variance that are commonly used to characterize V1 tuning \citep{ringach2002}. The index used here is defined by
\begin{equation}
\mathrm{OSI}
=
\frac{\Rpref-\Rorth}{\Rpref+\Rorth}.
\label{eq:osi}
\end{equation}
Because the rate-coded SCI states lie in $[-1,+1]$, the responses are rescaled to $[0,1]$ before computing Eq.~\eqref{eq:osi}. Under this convention, $\mathrm{OSI}\approx1$ indicates sharp orientation selectivity, whereas $\mathrm{OSI}\approx0$ indicates weak or absent selectivity.

%==============================================================================
\section{Results}
\label{sec:results}
%==============================================================================

The simulation study is organized into three principal experiments. The first two
revisit the orientation-selectivity and contrast-invariance behavior of the
Bayes--Markov model, pairing the original qualitative activity-pattern demonstrations
with a quantitative analysis based on the orientation selectivity index and controlled
noise. The third introduces a spiking realization of the SCI layer and evaluates it
through a battery of experiments spanning baseline selectivity, per-neuron heterogeneity,
multi-trial firing-rate stability, and single-neuron gain analysis. Throughout, the
qualitative (image-based) results show the spatial mechanism directly, while the
quantitative results give it a measurable form.

Unless otherwise noted, all experiments use $128\times128$ lattices with
$\beta=1.0$, $\gamma=2.0$, $T=0.15$, and the fully LGN-driven update rule of
Section~\ref{sec:model}.

%==============================================================================
\subsection{Simulation 1: Orientation Selectivity}
\label{sec:sim1}
%==============================================================================

\subsubsection{Qualitative orientation response}

A bar stimulus was presented at a range of orientations and propagated through the LGN
preprocessing stage and the SCI inference. For each orientation, the resulting activity
pattern was recorded as an input--LGN--SCI triptych. The SCI layer responds strongly when
the stimulus matches the model's preferred (horizontally tuned) orientation and is
suppressed otherwise, reproducing the orientation-selective behavior of the original
model across the principal orientations.

\begin{figure}[htpb]
\centering
\safeincludegraphics{width=0.82\linewidth}{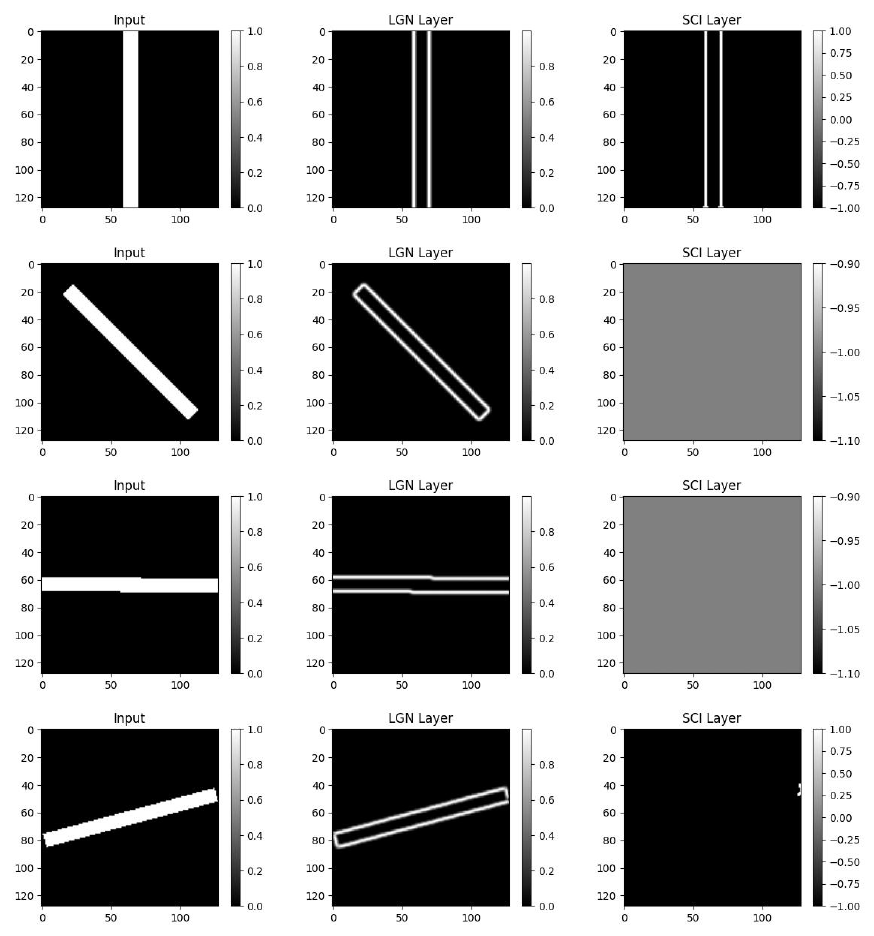}{7.0cm}
\caption{
Qualitative orientation response of the Bayes--Markov model. Each row shows the input
bar, the LGN activation field, and the inferred SCI activity at a given orientation. The
SCI layer is strongly active for the preferred orientation and suppressed for
non-preferred orientations.
}
\label{fig:sim1_patterns}
\end{figure}

\subsubsection{Quantitative tuning curve}

To give this behavior a quantitative form, a bar was presented at orientations from
$0^\circ$ to $180^\circ$ in increments of $10^\circ$, and the mean SCI activation inside
the bar region was recorded at each orientation. The resulting tuning curve peaks sharply
at the preferred orientation and is strongly suppressed at the orthogonal orientation,
yielding
\[
\mathrm{OSI}=1.0.
\]
Within this implementation, the result gives a quantitative form to the qualitative claim that the
inference produces a sharply selective orientation response from non-oriented LGN input.

\begin{figure}[htpb]
\centering
\safeincludegraphics{width=0.72\linewidth}{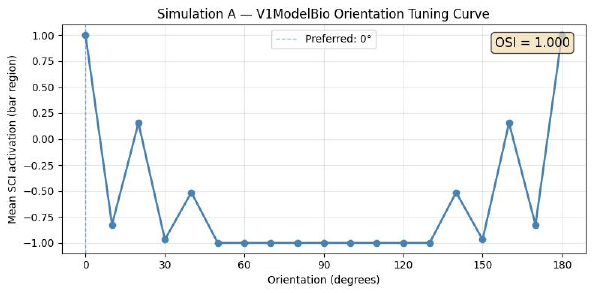}{7.0cm}
\caption{
Orientation tuning curve for the rate-coded SCI model (V1ModelBio). The response peaks at
the preferred orientation and is strongly suppressed at orthogonal angles, yielding
$\mathrm{OSI}=1.0$. The small secondary responses at oblique orientations are expected from
the $3\times3$ neighborhood, which naturally represents the principal orientations
$0^\circ$, $45^\circ$, $90^\circ$, and $135^\circ$.
}
\label{fig:simA_orientation}
\end{figure}

The combination of the qualitative pattern (Fig.~\ref{fig:sim1_patterns}) and the
quantitative curve (Fig.~\ref{fig:simA_orientation}) demonstrates orientation-selective behavior
both visually and numerically. Finer angular resolution at oblique orientations would
require larger neighborhoods or more elaborate clique configurations, consistent with the
spatial resolution of the original $3\times3$ formulation.

%==============================================================================
\subsection{Simulation 2: Contrast Invariance and Noise Robustness}
\label{sec:sim2}
%==============================================================================

\subsubsection{Qualitative contrast invariance}

A central property examined in the original model is the approximate invariance of orientation
selectivity to stimulus contrast, a major constraint in models of V1 tuning
\citep{sclar1982,troyer1998,ferster2000,marino2005}. To test this qualitatively, a bar of fixed orientation was presented at
several contrast levels, and the SCI activity pattern was recorded at each level. The
inferred SCI pattern is essentially unchanged across the tested contrast values, supporting the
interpretation that the local similarity-based inference depends primarily on the relative structure
of the LGN input rather than its absolute magnitude.

\begin{figure}[htpb]
\centering
\safeincludegraphics{width=0.82\linewidth}{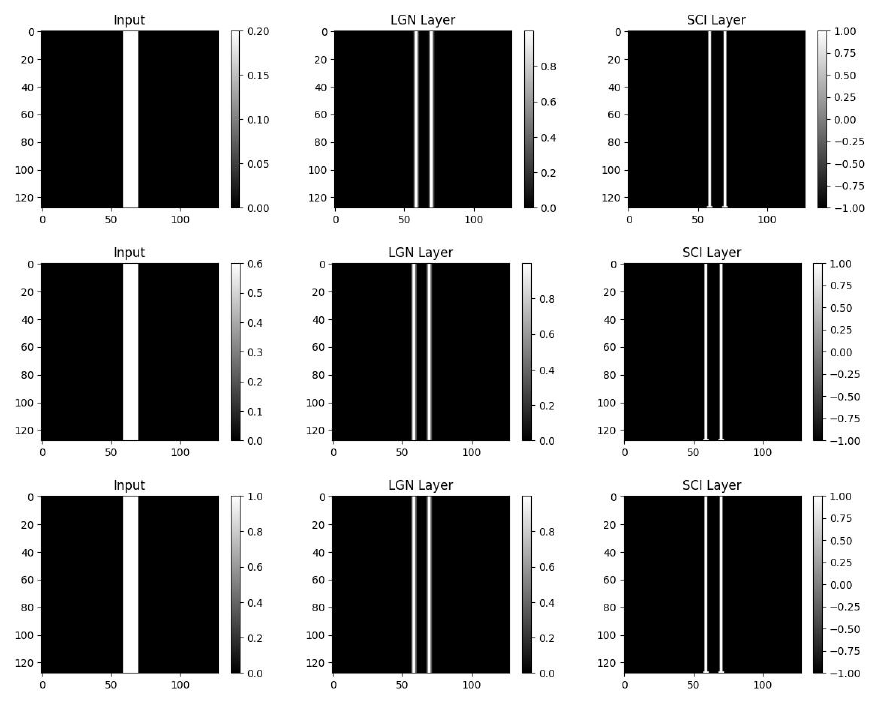}{7.0cm}
\caption{
Qualitative contrast invariance. The inferred SCI activity pattern is essentially
unchanged as stimulus contrast varies, illustrating that selectivity depends on relative
LGN structure rather than absolute activation magnitude.
}
\label{fig:sim2_contrast}
\end{figure}

\subsubsection{Quantitative noise robustness across orientations and model variants}

Robustness was then quantified by adding zero-mean Gaussian noise of standard deviation
$\sigma\in\{0,0.05,0.1,0.2,0.3,0.5\}$ to the LGN activation field and measuring the mean
response inside the stimulated bar region. To probe both the preferred response and the
stability of suppression, the sweep was repeated for bars at four orientations:
$0^\circ$ (preferred), $10^\circ$, $45^\circ$, and $90^\circ$ (orthogonal). At each noise
level an identical noise realization was applied to every model, including the spiking
models introduced in Simulation~3, so the comparison is fair across variants. Because the
rate-coded models report signed SCI activation in $[-1,+1]$ while the spiking models report
firing rate in spikes per second, the two families are shown in separate panels with their
natural response axes (Figs.~\ref{fig:sim2_noise_0}--\ref{fig:sim2_noise_90}).

At the preferred orientation (Fig.~\ref{fig:sim2_noise_0}), the response is preserved for
small perturbations: V1ModelBio remains at full activation ($+1.0$) up to
$\sigma\approx0.05$, falls to $\approx+0.56$ at $\sigma=0.1$, and is actively suppressed
(negative activation) for $\sigma\ge0.2$. The spiking models track this behavior in
firing-rate terms, declining from $\approx49$~Hz (LIF) and $\approx67$~Hz (HH) at low noise
to a few spikes per second beyond $\sigma\approx0.1$. This identifies a practical noise
regime, near $\sigma\approx0.1$, within which the local inference remains stable.

The three non-preferred orientations show that noise does not spuriously activate the
wrong orientation. At $90^\circ$ (Fig.~\ref{fig:sim2_noise_90}) the model is fully and
stably suppressed: V1ModelBio stays at $-1.0$ and both spiking models remain silent
($\approx0$~Hz) at every noise level. At $10^\circ$ (Fig.~\ref{fig:sim2_noise_10}) the
response is likewise suppressed throughout ($\approx-0.8$ rate-coded activation, $<10$~Hz
spiking), showing that even a small angular offset from the preferred orientation is
rejected. At $45^\circ$ (Fig.~\ref{fig:sim2_noise_45}) a small residual response is present
at zero noise (V1ModelBio $\approx-0.5$, HH $\approx19$~Hz), consistent with the weak
secondary oblique response expected from the $3\times3$ neighborhood
(Fig.~\ref{fig:simA_orientation}); this residual collapses by $\sigma\approx0.1$. In every
non-preferred case the response only moves further toward suppression as noise increases,
never toward activation.

\begin{figure}[htpb]
\centering
\safeincludegraphics{width=0.95\linewidth}{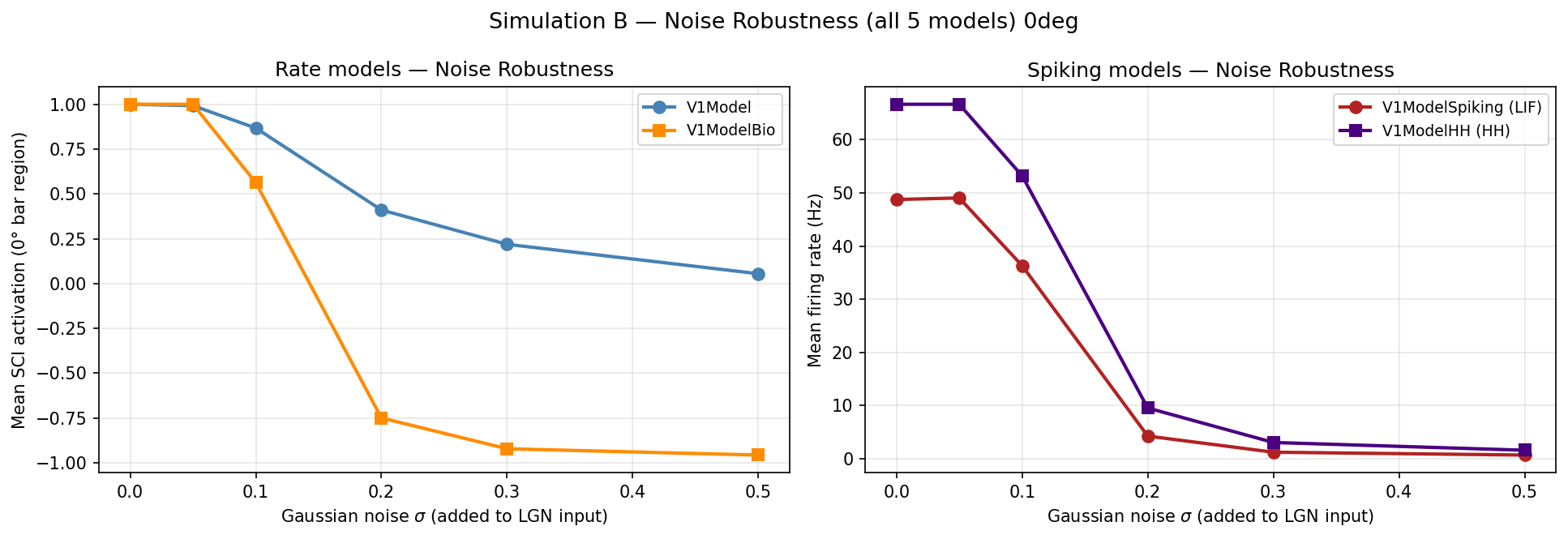}{6.0cm}
\caption{
Noise robustness at the preferred orientation ($0^\circ$). Left: rate-coded models (signed
SCI activation). Right: spiking models (mean firing rate). The preferred response is
preserved up to $\sigma\approx0.05$ and degrades beyond $\sigma\approx0.1$.
}
\label{fig:sim2_noise_0}
\end{figure}

\begin{figure}[htpb]
\centering
\safeincludegraphics{width=0.95\linewidth}{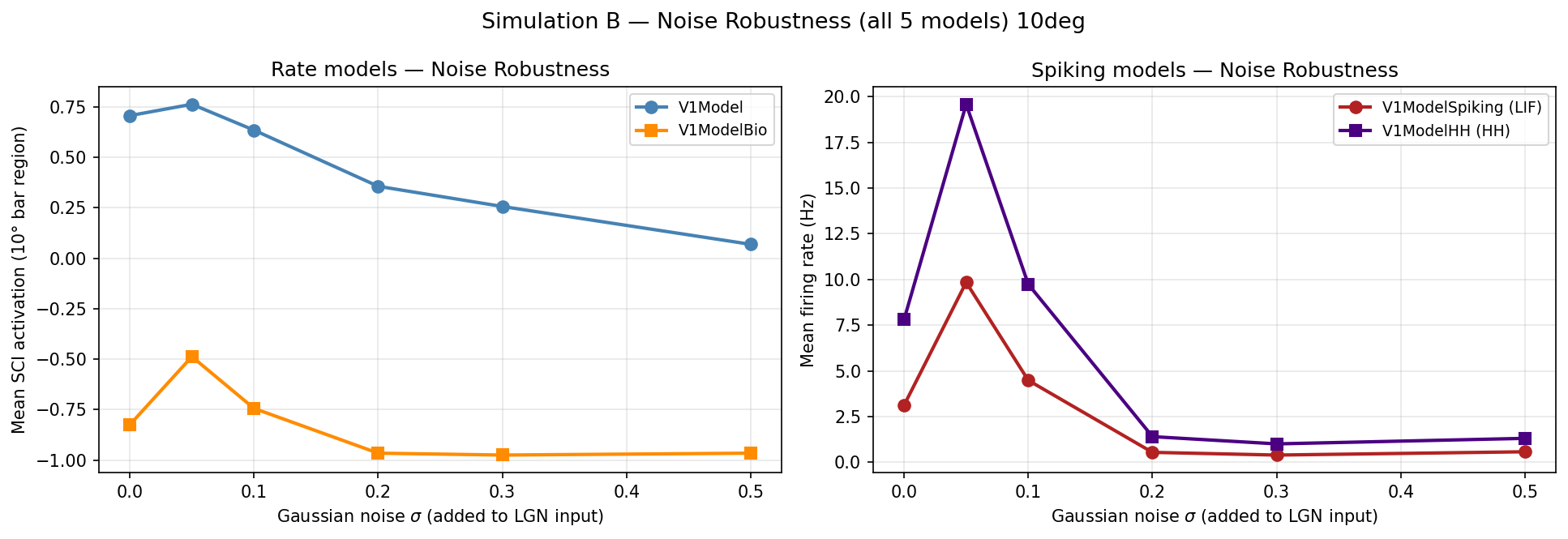}{6.0cm}
\caption{
Noise robustness at $10^\circ$. A small angular offset from the preferred orientation is
suppressed at every noise level, and the response never moves toward activation.
}
\label{fig:sim2_noise_10}
\end{figure}

\begin{figure}[htpb]
\centering
\safeincludegraphics{width=0.95\linewidth}{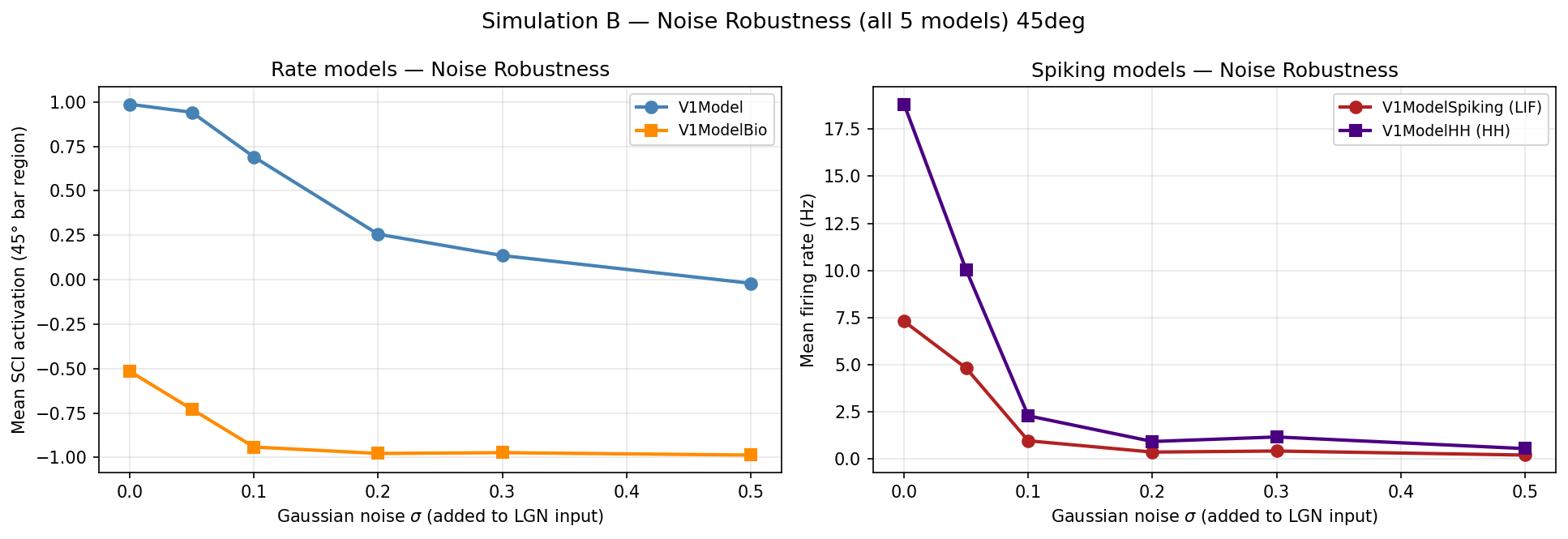}{6.0cm}
\caption{
Noise robustness at $45^\circ$. A weak secondary oblique response, expected from the
$3\times3$ neighborhood, is present at zero noise and collapses by $\sigma\approx0.1$.
}
\label{fig:sim2_noise_45}
\end{figure}

\begin{figure}[htpb]
\centering
\safeincludegraphics{width=0.95\linewidth}{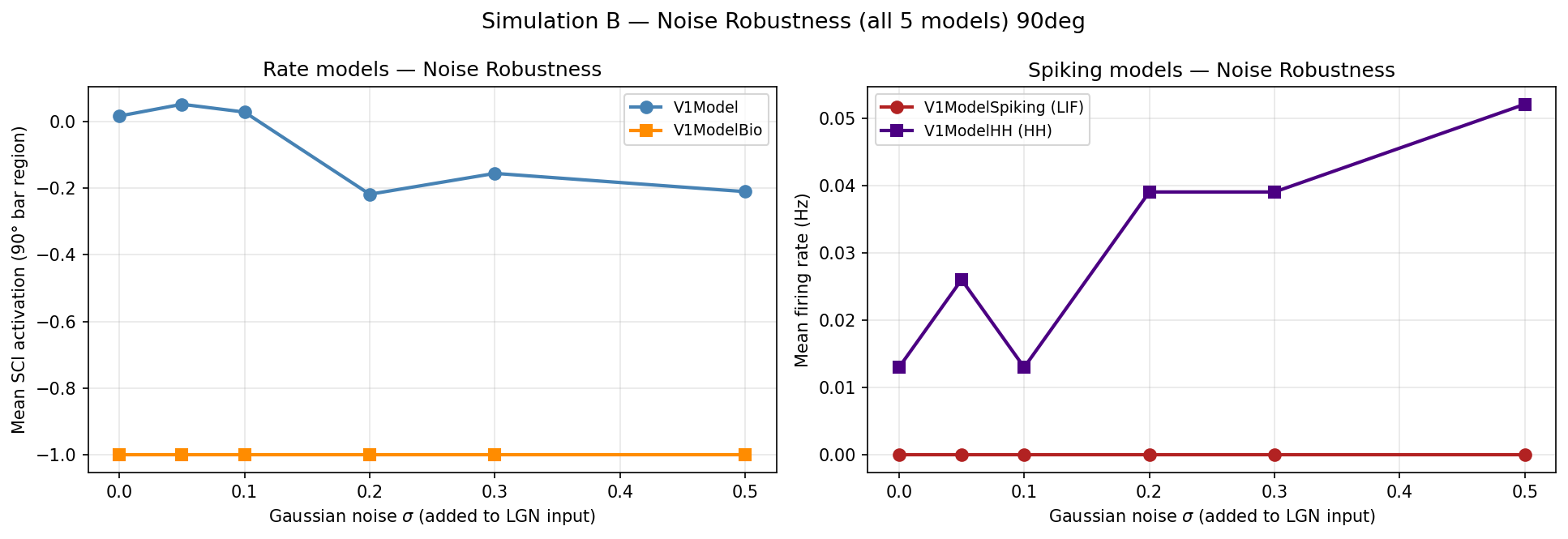}{6.0cm}
\caption{
Noise robustness at the orthogonal orientation ($90^\circ$). The model is fully and stably
suppressed (rate-coded activation $-1.0$; spiking rate $\approx0$~Hz) across all noise
levels.
}
\label{fig:sim2_noise_90}
\end{figure}

Together, the qualitative contrast analysis (Fig.~\ref{fig:sim2_contrast}) and the
orientation-resolved degradation curves
(Figs.~\ref{fig:sim2_noise_0}--\ref{fig:sim2_noise_90}) give a comparative form to
the original robustness claim within the tested simulation regime: the preferred response is contrast-invariant and tolerant of
moderate LGN noise, non-preferred suppression is maintained under noise, and only unbounded
perturbation destroys the selective response.

%==============================================================================
\subsection{Cross-Model Comparison}
\label{sec:simC}
%==============================================================================

The two rate-coded experiments above were carried out with the biologically motivated
implementation (V1ModelBio). Comparisons among feedforward, recurrent, balanced, and spiking models are common in the orientation-selectivity literature \citep{ursino2004,sadeh2014,sadeh2015,merkt2019,chariker2020}. To place this implementation in context, the model variants
were compared directly under a shared parameter setting
($\beta=1.0$, $\gamma=2.0$, $T=0.15$). Because all variants are evaluated with identical
input and nominal parameters, the comparison reflects their behavior under a common
operating condition rather than individually optimized performance.

Under this shared setting, V1ModelBio produced the clearest orientation-selective output
($\mathrm{OSI}=1.0$), whereas the baseline V1Model was more broadly tuned
($\mathrm{OSI}=0.33$); the side-by-side output maps are shown in
Fig.~\ref{fig:simC_outputs}. Both spiking realizations remained sharply selective under the
same drive, with $\mathrm{OSI}=1.0$ for the LIF model and $\mathrm{OSI}=0.998$ for the HH
model, and the full set of tuning curves is compared in Fig.~\ref{fig:simC_curves}. The
comparison should therefore be read as a shared-parameter analysis rather than an
individually optimized one.

\begin{figure}[htpb]
\centering
\safeincludegraphics{width=0.88\linewidth}{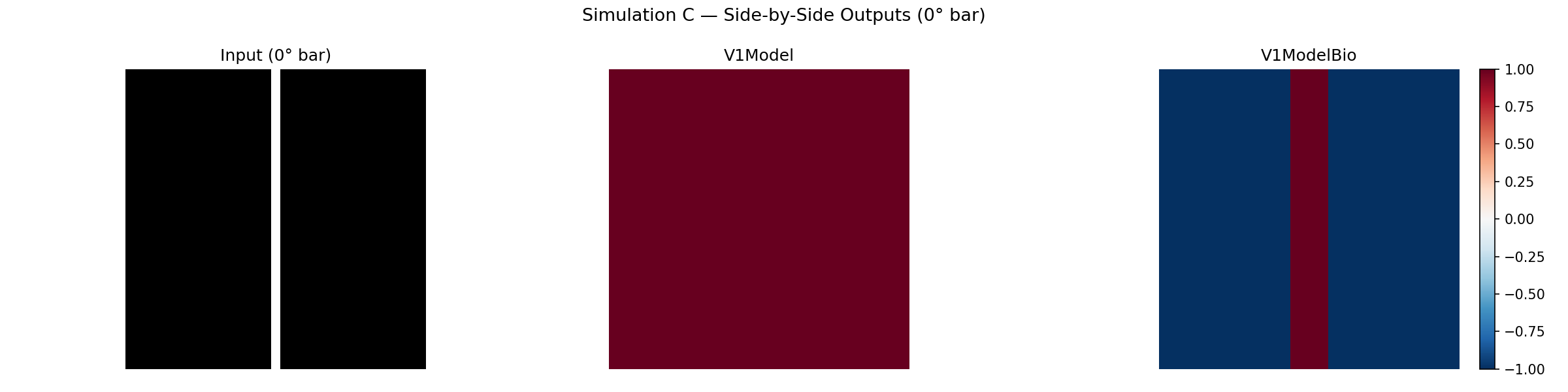}{7.0cm}
\caption{
Side-by-side SCI outputs from V1Model and V1ModelBio under identical input and shared
parameter values. V1ModelBio gives the sharpest selective response under this setting.
}
\label{fig:simC_outputs}
\end{figure}

\begin{figure}[htpb]
\centering
\safeincludegraphics{width=0.96\linewidth}{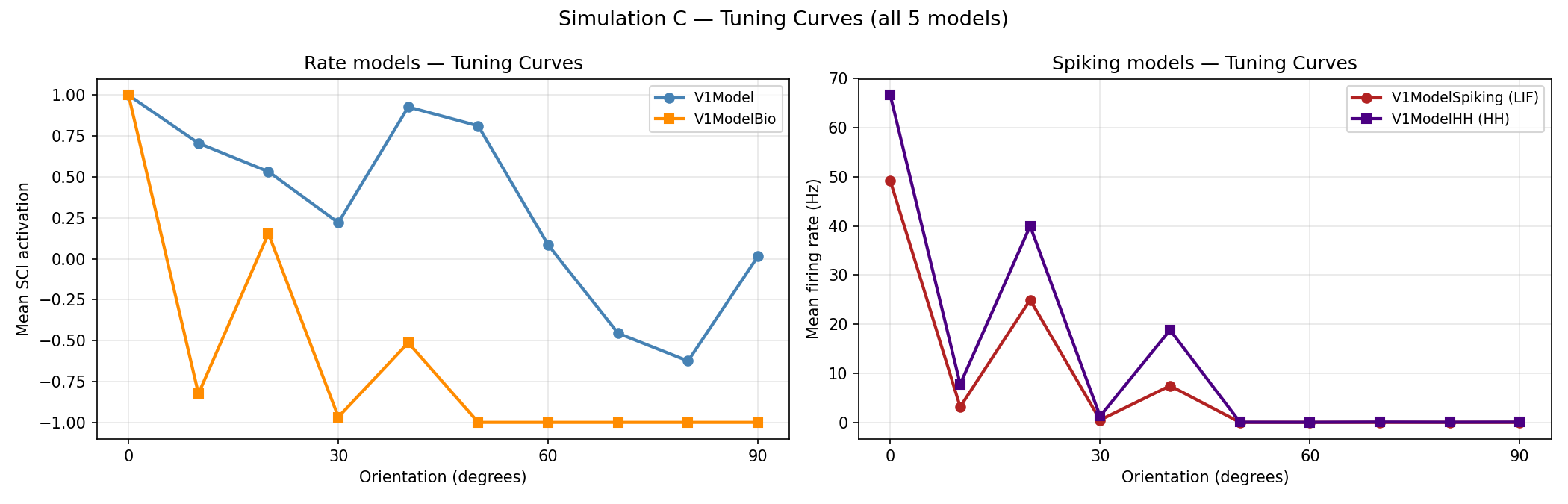}{7.0cm}
\caption{
Orientation tuning curves across the model variants. Left: rate-coded models (signed SCI
activation). Right: spiking models (firing rate). The biologically motivated rate model and
both spiking models are sharply selective, whereas the baseline model is more broadly
tuned.
}
\label{fig:simC_curves}
\end{figure}

%==============================================================================
\section{Simulation 3: Spiking SCI-Layer Implementation}
\label{sec:spiking}
%==============================================================================

The first two simulations use the original rate-coded SCI representation, in which each
cortical cell takes one of two values. This binary representation is natural for MAP
inference, but it imposes an artificial form of bistability at the rate-code level. The
third simulation asks whether the inferred SCI field can instead be expressed through a
biologically explicit temporal mechanism, and in particular whether Shirazi's own
prediction about spiking implementations holds.

\begin{prediction}[Shirazi's proposed temporal-coding interpretation]
\label{pred:shirazi}
If the rate-coded implementation is replaced by a temporal-coding implementation using
Hodgkin--Huxley-type spiking neurons, the artificial bistability of the rate code is
replaced by the biologically intrinsic, all-or-none bistability of action potentials
generated by voltage-gated sodium and potassium channels.
\end{prediction}

The spiking realization does not change the Bayes--Markov inference step; it is a downstream temporal readout of the inferred MAP field rather than a replacement of the inference algorithm. First, the rate-coded
MAP procedure produces an orientation-selective SCI field $s_{\mathrm{MAP}}$. This field is
then converted into an input current,
\begin{equation}
I_{\mathrm{LGN}}
=
I_{\mathrm{offset}}
+
I_{\mathrm{scale}}\,s_{\mathrm{MAP}},
\label{eq:spiking_current}
\end{equation}
which drives a grid of spiking neurons. The spiking output is summarized by firing rate in
spikes per second, a non-negative temporal realization of the same inferred SCI pattern.

\begin{algorithm}[htpb]
\caption{Spiking SCI-Layer Realization}
\label{alg:spiking}
\begin{algorithmic}[1]
\Procedure{SpikingSCI}{$Y$, neuron model $\mathcal{M}$, $\Delta t$, $t_{\max}$}
  \State \textbf{Input:} LGN field $Y$; neuron model $\mathcal{M}\in\{\text{LIF},\text{HH}\}$; step $\Delta t$; duration $t_{\max}$
  \State \textbf{Output:} firing-rate field $R=\{r_{ij}\}$
  \Statex
  \State $s_{\mathrm{MAP}} \gets \Call{InferSCI}{Y,\dots}$ \Comment{rate-coded MAP field (Alg.~\ref{alg:inference})}
  \State $I_{\mathrm{LGN}} \gets I_{\mathrm{offset}} + I_{\mathrm{scale}}\, s_{\mathrm{MAP}}$ \Comment{static drive, Eq.~\eqref{eq:spiking_current}}
  \State initialize membrane state $V$, spike counts $n_{ij}\gets 0$
  \For{$t \gets 0$ \textbf{to} $t_{\max}$ \textbf{step} $\Delta t$}
    \State $I_{\mathrm{syn}} \gets$ lateral synaptic current from recent neighbor spikes
    \State $I_{\mathrm{tot}} \gets I_{\mathrm{LGN}} + \gamma_{\mathrm{lat}}\,I_{\mathrm{syn}}$
    \State integrate $\mathcal{M}$ one step: update $V$ given $I_{\mathrm{tot}}$
           \Comment{Eq.~\eqref{eq:hh_voltage}--\eqref{eq:hh_gates} for HH}
    \For{each site with a threshold crossing}
      \State emit spike; $n_{ij} \gets n_{ij}+1$; reset $V_{ij}$
    \EndFor
  \EndFor
  \State $r_{ij} \gets n_{ij}\,/\,t_{\max}$ \Comment{firing rate (spikes per second)}
  \State \Return $R$
\EndProcedure
\end{algorithmic}
\end{algorithm}

\subsection{LIF and Hodgkin--Huxley Models}

Two spiking neuron models were implemented. The first is a leaky integrate-and-fire (LIF)
model, included as a simple comparison baseline following the classical integrate-and-fire
tradition \citep{lapicque1907}. The second is a Hodgkin--Huxley (HH) model
\citep{hodgkin1952}, included because it is the biophysically explicit mechanism named in
Prediction~\ref{pred:shirazi}, generating all-or-none action potentials through
voltage-gated sodium and potassium channels.

The LIF model is defined by a passive membrane equation with threshold-and-reset spiking.
The HH model is described by the standard conductance-based equations
\begin{equation}
C_m\frac{dV}{dt}
=
I
-
g_{\mathrm{Na}}m^3h(V-E_{\mathrm{Na}})
-
g_{\mathrm{K}}n^4(V-E_{\mathrm{K}})
-
g_{\mathrm{L}}(V-E_{\mathrm{L}}),
\label{eq:hh_voltage}
\end{equation}
together with gating-variable dynamics
\begin{equation}
\frac{dq}{dt}
=
\alpha_q(V)(1-q)-\beta_q(V)q,
\qquad
q\in\{m,h,n\}.
\label{eq:hh_gates}
\end{equation}
The HH model therefore provides a mechanistic spiking realization of the SCI field rather
than an imposed binary firing-rate code.

\subsection{Experiment 1: Baseline spiking orientation selectivity}

Both spiking implementations preserved the sharp orientation selectivity of the rate-coded
MAP field. Driven by the preferred ($0^\circ$) stimulus, the LIF model fired at
approximately $49$~Hz and the HH model at approximately $67$~Hz, while both were silent at
the orthogonal ($90^\circ$) orientation, yielding $\mathrm{OSI}=1.0$ for the LIF model and
$\mathrm{OSI}=0.999$ for the HH model. No cell in the preferred-orientation population
exceeded $100$~Hz, showing that the injected drive remains bounded in the tested simulation
rather than producing runaway activity. Figures~\ref{fig:hh_maps_raster}
and~\ref{fig:lif_maps_raster} show the corresponding firing-rate maps and spike rasters:
activity is confined to the preferred bar, and the raster displays temporally explicit,
all-or-none action potentials rather than the imposed $\pm1$ values of the rate code.

\begin{figure}[htpb]
\centering
\safeincludegraphics{width=0.86\linewidth}{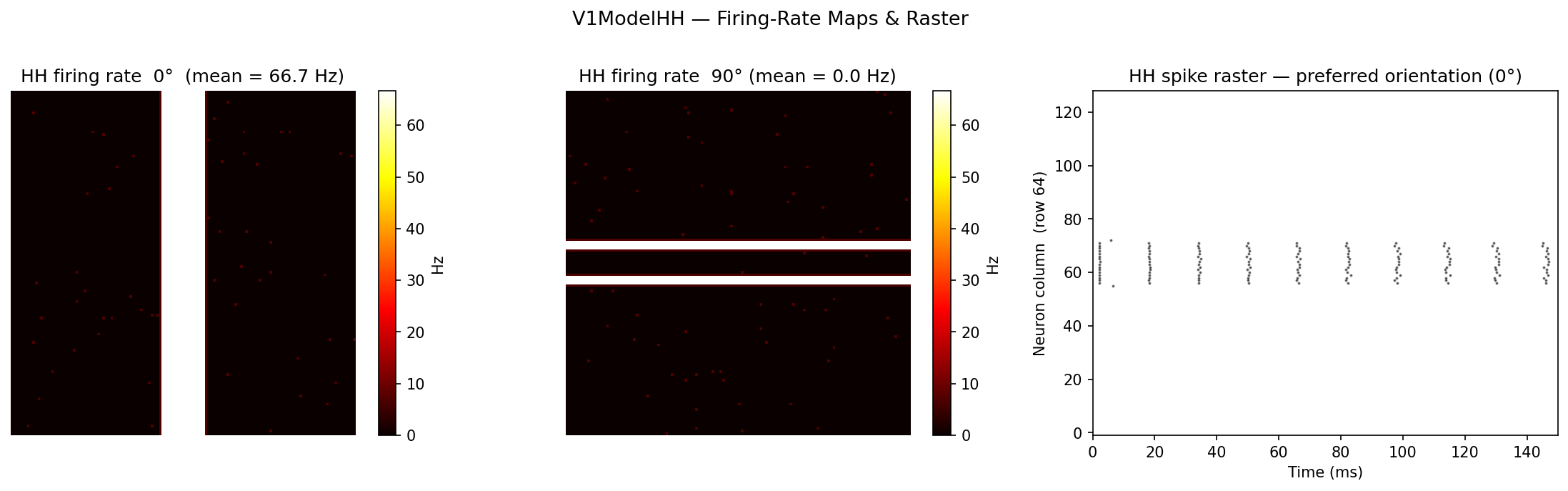}{7.0cm}
\caption{
Hodgkin--Huxley spiking implementation of the SCI layer. The HH model fires strongly for
the preferred orientation and remains near silent for the orthogonal orientation while
producing temporally explicit spike trains. This provides a proof-of-concept spiking realization
consistent with the all-or-none behavior described in Prediction~\ref{pred:shirazi}.
}
\label{fig:hh_maps_raster}
\end{figure}

\begin{figure}[htpb]
\centering
\safeincludegraphics{width=0.86\linewidth}{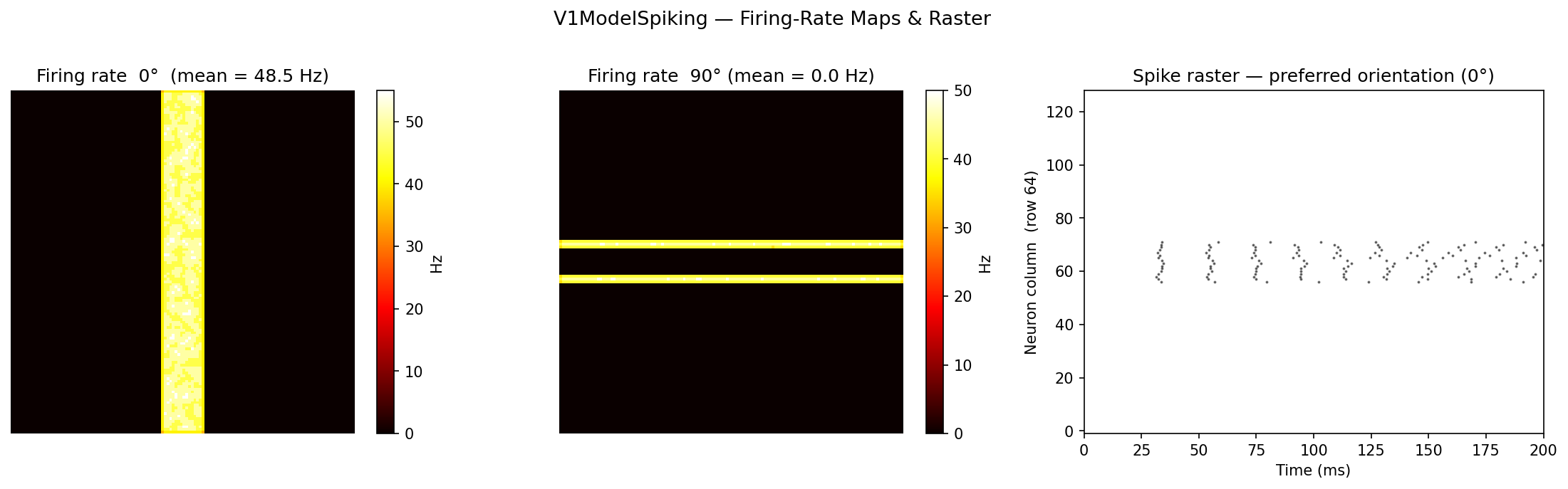}{7.0cm}
\caption{
LIF spiking implementation of the SCI layer, included as a comparison baseline. Like the HH
model, the LIF model confines activity to the preferred orientation and produces temporally
explicit spike trains.
}
\label{fig:lif_maps_raster}
\end{figure}

\subsection{Experiment 2: Per-neuron heterogeneity and multi-trial firing-rate stability}

Prediction~\ref{pred:shirazi} asks whether the bistable SCI field, once expressed through
spiking dynamics, remains all-or-none under realistic variability rather than as a
single-run artifact. Two sources of variability were therefore introduced. First, fixed
per-neuron heterogeneity was added by multiplying each neuron's MAP-derived input current by
a static random gain drawn from $\mathcal{U}(0.85,1.15)$. Second, the measurement was
repeated over ten independent stochastic-noise realizations, and the coefficient of
variation (CV) of the active-cell firing rates in the preferred-orientation region was
computed both per trial and pooled across all trials. The results are summarized in
Table~\ref{tab:spiking_results}.

\begin{table}[htpb]
\centering
\caption{
Orientation selectivity and multi-trial firing-rate variability for the spiking SCI
implementations. OSI is measured from the orientation tuning sweep. The coefficient of
variation (CV) is computed over active-cell firing rates in the preferred-orientation bar
region: the pooled CV aggregates all active cells across ten independent noise realizations,
and the per-trial CV is the mean $\pm$ standard deviation of the single-trial CV values.
}
\label{tab:spiking_results}
\renewcommand{\arraystretch}{1.15}
\begin{tabular}{lccc}
\toprule
\textbf{Model} & \textbf{OSI} & \textbf{Pooled CV ($10$ trials)} & \textbf{Per-trial CV} \\
\midrule
LIF baseline        & 1.000 & 0.061 & $0.061 \pm 0.001$ \\
LIF + heterogeneity & 1.000 & 0.101 & $0.101 \pm 0.001$ \\
HH baseline         & 0.999 & 0.000 & $0.000 \pm 0.000$ \\
HH + heterogeneity  & 0.973 & 0.005 & $0.004 \pm 0.003$ \\
\bottomrule
\end{tabular}
\end{table}

The two neuron models behave qualitatively differently. The LIF model remains orientation
selective but grades under heterogeneity: its pooled CV rises from $0.061$ (baseline) to
$0.101$ (heterogeneous), and the per-trial CV is tight ($0.061\pm0.001$ and
$0.101\pm0.001$), so the grading is a stable property of the model rather than
trial-to-trial noise. The HH model instead remains all-or-none: its pooled CV is $0.000$ at
baseline and only $0.005$ under heterogeneity, with per-trial CV $0.004\pm0.003$. In the HH
implementation active cells therefore fire at a characteristic rate or remain silent, and
this bistability persists across independent noise realizations
(Fig.~\ref{fig:spiking_distributions}), consistent with the
action-potential bistability of Prediction~\ref{pred:shirazi}.

\begin{figure}[htpb]
\centering
\safeincludegraphics{width=0.88\linewidth}{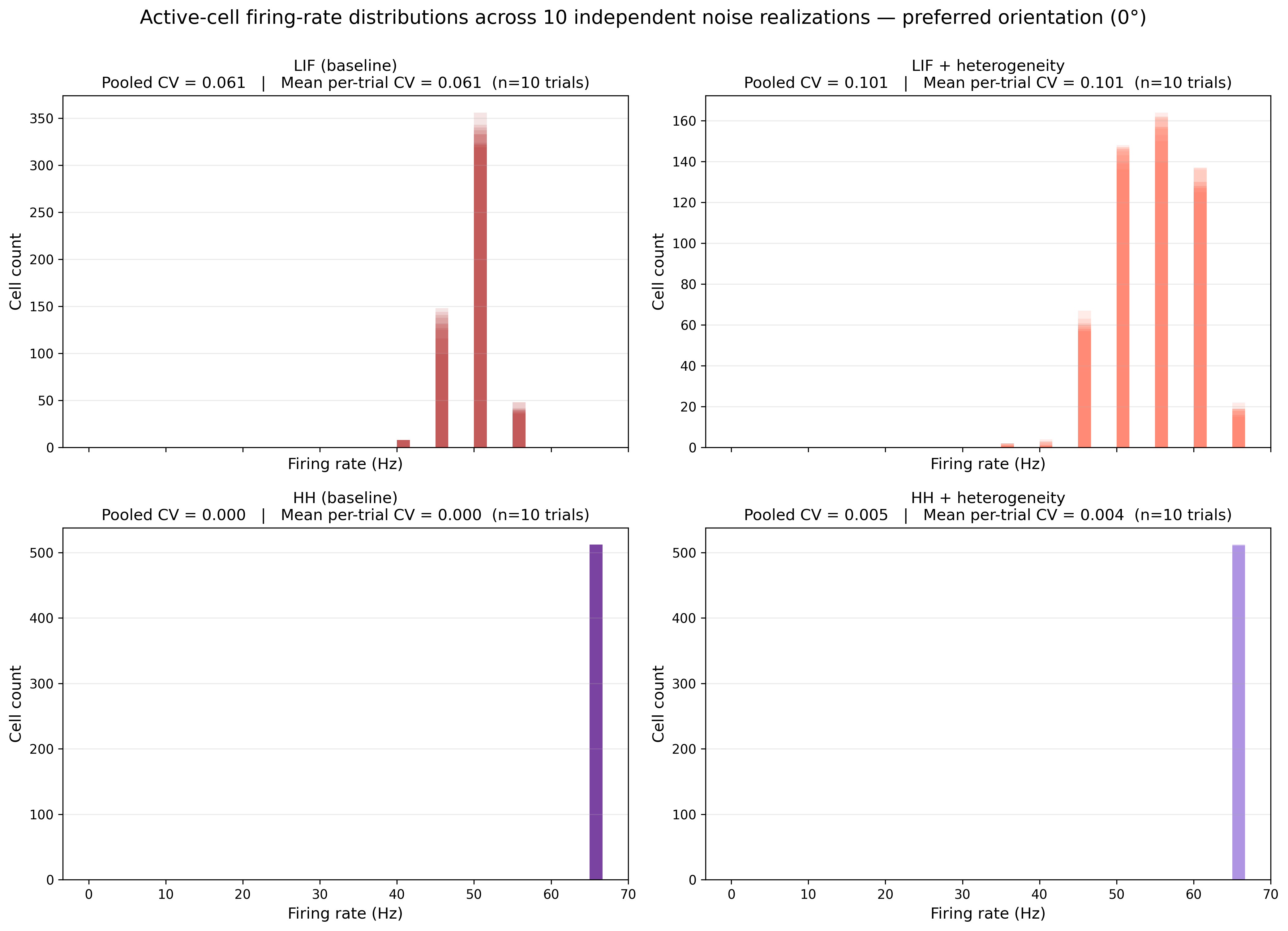}{7.5cm}
\caption{
Active-cell firing-rate distributions across ten independent noise realizations (preferred
orientation, $0^\circ$). Faint curves are individual trials; each panel title reports the
pooled and mean per-trial CV. The HH distributions remain sharply concentrated at a single
rate (all-or-none firing) across trials, whereas the LIF distributions broaden under input
heterogeneity. The concentrated HH distribution is consistent with the modeled prediction.
}
\label{fig:spiking_distributions}
\end{figure}

\subsection{Experiment 3: Single-neuron \texorpdfstring{$f$--$I$}{f--I} analysis}

The contrast between the two neuron models is explained by their single-neuron
frequency--current ($f$--$I$) curves (Fig.~\ref{fig:fI_compare}). The LIF neuron has a
smooth, continuously rising $f$--$I$ curve: over the heterogeneity drive range
($0.935$--$1.265$ in dimensionless input units) its firing rate spans the full interval
$0$--$35$~Hz, a spread of $35$~Hz, so a $\pm15\%$ gain variation maps directly into a
proportional spread of firing rates. The HH neuron instead has an abrupt spike-initiation
threshold followed by a near-plateau firing regime: over the corresponding heterogeneity
drive range ($6.8$--$9.2~\mu\mathrm{A/cm}^2$) its firing rate varies only between $60$ and
$66.7$~Hz, a spread of $6.7$~Hz. The gain heterogeneity therefore falls almost entirely on
the HH plateau, so cells receiving different drives fire at nearly the same rate. The
all-or-none behavior of the HH population is thus a direct consequence of its voltage-gated
channel dynamics, supporting the behavior described in Prediction~\ref{pred:shirazi}.

\begin{figure}[htpb]
\centering
\safeincludegraphics{width=0.92\linewidth}{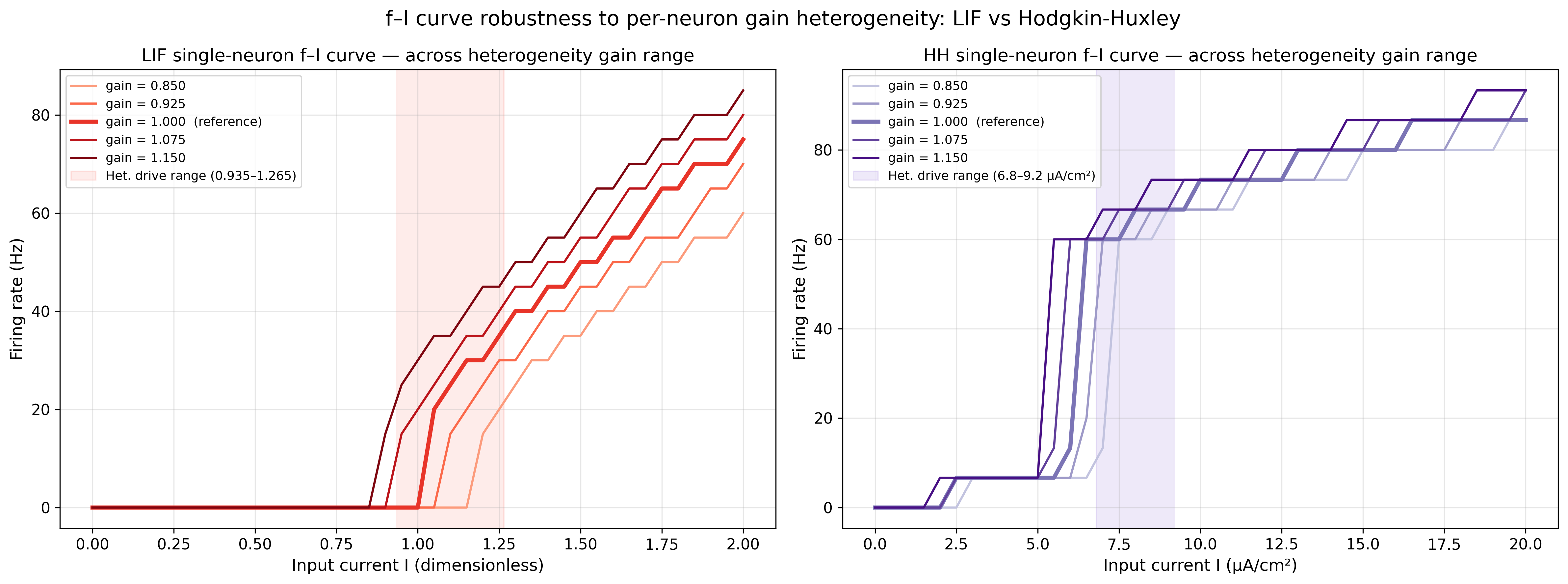}{6.5cm}
\caption{
Single-neuron frequency--current ($f$--$I$) curves across the per-neuron gain heterogeneity
range. Left: the LIF neuron has a smooth, continuously rising curve, so the shaded
heterogeneity drive range maps to a wide firing-rate spread ($\approx35$~Hz). Right: the HH
neuron has an abrupt onset and a near-plateau regime, so the shaded heterogeneity drive
range maps to a narrow spread ($\approx6.7$~Hz). This explains why the HH population
maintains all-or-none firing while the LIF population grades.
}
\label{fig:fI_compare}
\end{figure}

%==============================================================================
\section{Discussion}
\label{sec:discussion}
%==============================================================================

The simulations provide a quantitative re-examination of the original Bayes--Markov model. The main conclusion is that the core mechanism proposed by \citet{shirazi2004} remains computationally coherent under the numerical experiments considered here.

First, Simulation~1 demonstrates that the model produces a sharply tuned response, pairing the qualitative orientation patterns with a tuning curve that attains $\mathrm{OSI}=1.0$. This is consistent with the central claim that orientation-selective inhibition can emerge from non-oriented LGN inputs through local probabilistic inference. Second, Simulation~2 quantifies contrast invariance and robustness, identifying a practical transition near $\sigma=0.1$ within the tested noise scale and showing the corresponding degradation across all model variants. The cross-model comparison further shows that the fully LGN-driven biologically motivated implementation produces the clearest selective behavior under the shared parameters. This comparison is complementary to broader work showing that selectivity may arise or propagate through multiple combinations of feedforward structure, recurrent balance, and spiking dynamics \citep{benyishai1995,somers1995,marino2005,sadeh2015,merkt2019,chariker2020}.

The spiking implementation provides an additional biological interpretation. The original MAP formulation produces a binary SCI field. In the HH implementation, this binary field is not merely imposed as an artificial rate code; it is expressed through all-or-none action-potential generation. The LIF comparison is useful because it shows that simpler passive spiking dynamics can preserve orientation selectivity but may produce graded firing-rate variability under heterogeneity. The HH model, in contrast, preserves all-or-none firing because of its intrinsic voltage-gated channel dynamics. Across ten independent noise realizations the HH firing-rate distribution remained essentially degenerate (pooled $\mathrm{CV}\le0.005$), showing that this all-or-none behavior is stable across the tested realizations rather than being restricted to a single run.

Several limitations remain. The current experiments focus on controlled synthetic stimuli rather than large-scale natural image benchmarks. The $3\times3$ neighborhood naturally favors a small set of principal orientations, and finer orientation resolution would require larger neighborhoods or more detailed clique structures. These limitations define the scope of the present study, which is to re-implement and quantitatively evaluate the original Bayes--Markov theory rather than replace it with a new model.

%==============================================================================
\section{Conclusion}
\label{sec:conclusion}
%==============================================================================

This paper revisited Shirazi's Bayes--Markov model of cortical orientation selectivity and provided a computational re-implementation with quantitative simulation validation. The original theory formulates the emergence of orientation-selective inhibition as a MAP estimation problem over a two-layer hierarchical Markov random field, solved by a local parallel relaxation algorithm. We preserved this formulation and focused on making its behavior measurable through systematic simulations.

The results support the main qualitative predictions of the original model. The reimplemented model produces sharp orientation tuning, achieves $\mathrm{OSI}=1.0$ in the biologically motivated implementation, and remains robust to moderate LGN noise. The vectorized implementation makes these simulations computationally feasible without changing the underlying local clique operations. Finally, the spiking SCI-layer implementation shows that the inferred orientation-selective field can be expressed through temporally explicit neural dynamics. The HH model provides a biologically interpretable proof-of-concept all-or-none spiking realization that is consistent with Shirazi's proposed temporal-coding interpretation.

%==============================================================================
% Declarations (Required by Springer)
%==============================================================================
\section*{Declarations}

\begin{itemize}[leftmargin=*]
    \item \textbf{Funding:} The authors declare that no funds, grants, or other support were received during the preparation of this manuscript.
    \item \textbf{Conflict of interest/Competing interests:} The authors have no relevant financial or non-financial interests to disclose.
    \item \textbf{Ethics approval:} Not applicable.
    \item \textbf{Consent to participate:} Not applicable.
    \item \textbf{Consent for publication:} Not applicable.
    \item \textbf{Availability of data and materials:} No external experimental dataset was used. All reported results were generated from synthetic stimuli and numerical simulations. The resulting simulation outputs are available from the corresponding author upon reasonable request.
    \item \textbf{Code availability:} The custom simulation code, parameter configurations, and scripts used to generate the reported results and figures are available from the corresponding author upon reasonable request.
    \item \textbf{Authors' contributions:} A.M., M.S., F.S., and H.P. conceptualized the study. A.M. and M.S. implemented the computational models and performed the simulations. F.S. contributed to the theoretical framework analysis. H.P. supervised the project. All authors contributed to writing and reviewing the manuscript.
\end{itemize}

%==============================================================================
% Appendix
%==============================================================================
\begin{appendices}
\section{Proof of Clique Energy Emulation}
\label{app:proof}

The SCI clique energy in Eq.~\eqref{eq:esci} appears to require direct access to the abstract SCI states $\Xlow$ and $\Xup$. A neuromorphic implementation should avoid such external reference signals. The following lemma states how local hypothesis-testing cells can emulate the same energy contribution through lateral interactions.

\begin{lemma}
\label{lem:emulation}
Let $z^{Q}_{C'}(n)$ and $z^{A}_{C'}(n)$ denote the activity levels of the quenching cell $Z^{Q}_{C'}$ and activating cell $Z^{A}_{C'}$ at iteration $n$, and let $\Lambda_z>0$ be their maximum activity level. Then the SCI clique energy can be represented locally as
\begin{align}
E(x_{ij}=\Xup,\xtilde_{C'}(n);\gamma(n))
&=
-\gamma(n)\frac{z^{Q}_{C'}(n)}{\Lambda_z},
\label{eq:lem1}
\\
E(x_{ij}=\Xlow,\xtilde_{C'}(n);\gamma(n))
&=
-\gamma(n)\frac{z^{A}_{C'}(n)}{\Lambda_z}.
\label{eq:lem2}
\end{align}
\end{lemma}

\begin{proof}
From Eq.~\eqref{eq:esci}, the SCI energy is
\[
E(x_{ij},\xtilde_{C'}(n);\gamma(n))
=
-\gamma(n)I(x_{ij},\xtilde_{C'}(n)).
\]
Therefore,
\[
E(x_{ij}=\Xup,\xtilde_{C'}(n);\gamma(n))
=
\begin{cases}
-\gamma(n), & \xtilde_{C'}(n)=\Xup,\\
0, & \xtilde_{C'}(n)\ne\Xup.
\end{cases}
\]
In the hypothesis-testing circuit, the quenching and activating cells are assumed to form a winner-take-all pair. Thus, at a stable local decision, either
\[
z^Q_{C'}(n)=\Lambda_z,\qquad z^A_{C'}(n)=0,
\]
or
\[
z^Q_{C'}(n)=0,\qquad z^A_{C'}(n)=\Lambda_z.
\]
If the neighboring SCI state corresponds to $\Xup$, the quenching cell represents that local state and $z^Q_{C'}(n)=\Lambda_z$. Hence
\[
-\gamma(n)\frac{z^Q_{C'}(n)}{\Lambda_z}
=
-\gamma(n).
\]
If the neighboring state does not correspond to $\Xup$, then $z^Q_{C'}(n)=0$, and the same expression becomes zero. This proves Eq.~\eqref{eq:lem1}. The proof of Eq.~\eqref{eq:lem2} is identical after exchanging the roles of the activating and quenching cells.
\end{proof}
\end{appendices}

%==============================================================================
% References
%==============================================================================

\end{document}